\documentclass[journal,draftcls,onecolumn,12pt,twoside]{IEEEtranTCOM}
\renewcommand{\baselinestretch}{1.62}  
\usepackage{moreverb}
\usepackage{epsfig}
\usepackage{amsmath,amssymb,amsthm,mathrsfs,amsfonts,dsfont}
\usepackage{adjustbox,lipsum}
\usepackage{algorithm,algorithmic}
\usepackage{amsfonts}
\usepackage{amssymb}
\usepackage{amsmath}
\usepackage{amsthm}
\usepackage{multirow}
\usepackage{setspace}
\usepackage{rotating}
\usepackage{graphicx}
\usepackage{url}
\usepackage{tabularx}
\usepackage{array}
\usepackage{anyfontsize}
\usepackage{color,soul}
\usepackage{graphicx,dblfloatfix}
\usepackage{epstopdf}
\usepackage{blindtext}
\usepackage{amsmath}
\usepackage{amsthm,amssymb,amsmath,bm}
\usepackage{subfig}
\usepackage{amsfonts}
\usepackage{epsfig}
\usepackage{amssymb}
\usepackage{amsmath}
\usepackage{cite}
\usepackage{graphicx}
\usepackage{fancyhdr}
\usepackage[font={small}]{caption}
\usepackage{tabularx}
\usepackage{tcolorbox}
\usepackage{cite}
\usepackage{setspace}

\newtheorem{theorem}{Theorem}

\newtheorem{rem}{Remark}

\allowdisplaybreaks
\begin{document}
	
	\title{Moment Generating Function of the AoI in { a Two-Source System}    With Packet Management
		\thanks{ This research has been financially supported by the Infotech Oulu, the Academy of Finland (grant 323698), and Academy of Finland 6Genesis Flagship (grant 318927). M. Codreanu would like to acknowledge the support of the European Union's Horizon 2020 research and innovation programme under the Marie Sk\l{}odowska-Curie Grant Agreement No. 793402 (COMPRESS NETS). The work of M. Leinonen has also been financially supported in part by the Academy of Finland (grant 319485). 
		M. Moltafet would like to acknowledge the support of Finnish Foundation for Technology Promotion,  HPY Research Foundation, Riitta ja Jorma J. Takanen Foundation, and Nokia Foundation.}
	\thanks{Mohammad Moltafet and Markus Leinonen are with the Centre
		for Wireless Communications--Radio Technologies, University of Oulu,
		90014 Oulu, Finland (e-mail: mohammad.moltafet@oulu.fi; markus.leinonen@oulu.fi), and
		Marian Codreanu is with Department of Science and Technology, Link\"{o}ping University, Sweden (e-mail: marian.codreanu@liu.se)
	}
\author{
		Mohammad~Moltafet, Markus~Leinonen, and Marian~Codreanu
}
}
	\maketitle

\begin{abstract}
We consider a status update system consisting of two independent sources and one server in which packets of each source are generated according to the Poisson process and packets are served according to an exponentially distributed service time.  We derive the moment generating function (MGF) of the  age of information  (AoI) for each source in the system by using the stochastic hybrid systems (SHS) under two existing source-aware packet management policies which we term  self-preemptive and non-preemptive policies. 
In the both policies, the system (i.e., the waiting queue and the server)  can contain at most two packets, one packet of each source;  when the server is busy and a new packet arrives, the possible packet of the same source in the waiting queue is replaced by the fresh packet. The main difference between the policies is that in the self-preemptive policy, the packet under service is replaced upon the arrival of a new packet from the same source, whereas in the non-preemptive policy, this new arriving packet is blocked and cleared.
%
We use the derived MGF to find the first and second moments of the AoI and show the importance of higher moments.

\end{abstract}

\section{Introduction}\label{Introduction}
In  low-latency cyber-physical system applications information freshness is critical. In such systems, various sensors are assigned to generate status update packets of various  real-world physical processes. These  packets are transmitted through a network to a sink.  Awareness of the sensors' state needs to be as timely as possible. Recently,  the age of information (AoI) was proposed to measure the information freshens  at the sink \cite{6195689,8469047,9099557}. 
If  at a time instant $t$, the most recently received status update packet contains the time stamp (i.e., the time when status update was generated) $U(t)$, AoI is defined
as the random process $\Delta(t)=t-U(t)$. Thus, the AoI measures for each sensor the time elapsed since the last received status update packet was generated.

The seminal work \cite{8469047} introduced a powerful technique, called \textit{stochastic hybrid systems} (SHS), to calculate the average AoI.
%
In \cite{9103131}, the authors extended the SHS analysis to calculate the moment generating function (MGF) of the AoI. 
 The SHS technique has  been used to analyze the AoI in various queueing models \cite{8437591,8406966,8437907,9013935,9048914,9007478,moltafet2020sourceawareage,9174099}.
 In  our work \cite{moltafet2020sourceawareage}, we derived the average AoI under three proposed source-aware packet management policies; these were shown to result in low AoI and high fairness among the sources.


 In this letter, we 
 extend our AoI analysis in \cite{moltafet2020sourceawareage} by calculating the MGF of the AoI for each source under the two source-aware packet management policies proposed in \cite{moltafet2020sourceawareage}
   which we term  \textit{self-preemptive} and  \textit{non-preemptive} policies.  
   {The main motivation behind considering these two policies is that the self-preemptive policy provides the lowest sum average AoI and the non-preemptive policy provides the fairest situation in the system among the policies studied in  \cite{moltafet2020sourceawareage}.}
   In both policies, the system (i.e., the waiting queue and the server)  can contain at most two packets, one packet of each source; and when the server is busy and a new packet arrives, the possible packet of the same source in the waiting queue is replaced by the fresh packet. The main difference between the policies is that in the self-preemptive policy, the packet under service is replaced upon the arrival of a new packet from the same source, whereas in the non-preemptive policy, this new arriving packet is blocked and cleared.
  By using the derived MGF, we derive the first and second moments of the AoI and show the importance of higher moments.


 \section{System Model}\label{System Model}
We  consider a status update system consisting of two sources\footnote{{We consider two sources for simplicity of presentation; the same methodology as used in this letter can be applied for more than two sources. However, the complexity of the calculations increases exponentially with the number of sources.}}  and  one server {as illustrated in Fig. \ref{MGFAoIs}.}
Each source observes a random process at random time instants and the measured value of the monitored process is transmitted as a status update packet.  Each status update packet contains the measured value of the monitored process and a time stamp representing the time when the sample was generated. 
{We assume that  the sources  generate their packets according to the Poisson process with rates  $\lambda_c,~c\in \{1,2\},$ and the server serves the packets according to an exponentially distributed service time with rate $\mu$.}

	\begin{figure}[t]
		\centering
	\includegraphics[width=0.75\linewidth]{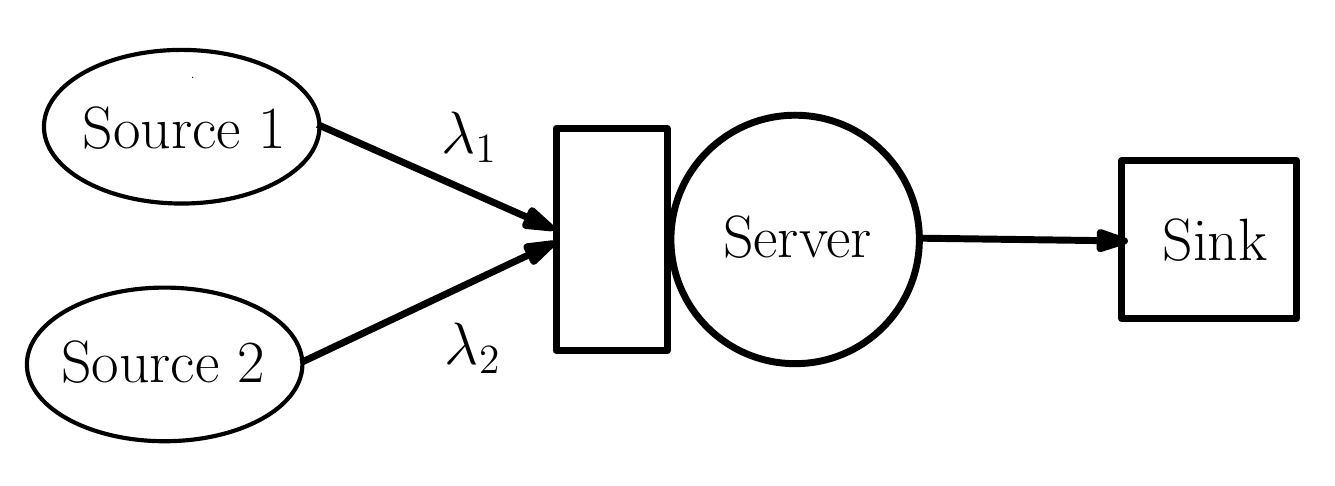}
	\caption{The considered status update system.}
		\vspace{-6mm}			
	\label{MGFAoIs}
\end{figure}

 Let  ${\rho_c={\lambda_c}/{\mu},~c\in \{1,2\},}$  be the  load of  source $ c$. Since packets of each source are generated  according to the Poisson process and the sources are independent, the  packet generation in the system follows the Poisson process with rate
 ${\lambda=\lambda_1+\lambda_2}$, 
%
%
and  the overall load in the system is
$\rho={\lambda}/{\mu}$.   Let $\Delta_c,~c\in \{1,2\},$ be the average AoI of source $c$.

\vspace{-.25cm}
\subsection{Packet Management Policies}\label{II-A}
 In both  packet management policies, the system   can contain at most two packets, one packet of each source; and when the server is busy and a new packet arrives, the possible packet of the same source in the waiting queue is replaced by the fresh packet. The main difference between the policies is that in the self-preemptive policy, the packet under service is replaced upon the arrival of a new packet from the same source, whereas in the non-preemptive policy, this new arriving packet is blocked and cleared.

\subsection{Summary of the Main Results}
In this paper, we derive the MGF of the AoI for each source under the self-preemptive and non-preemptive policies which are summarized by the following two theorems. { Note that while we derive the MGF for source 1, the derived results can be straightforwardly used to compute the MGF of AoI of source 2 by interchanging $ \rho_1 $ and $ \rho_2 $.  }
\begin{theorem}\label{MGFtheo1}
	The MGF of the  AoI of source 1 under the self-preemptive  policy is given as
	\begin{align}\label{mnb10}
	&M_{\Delta_1}(s)=\dfrac{\rho_1}{2\rho_1\rho_2+\rho+1}\times\left[\dfrac{\rho_2^2(1-\bar{s})^2+2\rho_2(1-\bar{s})^3+(1-\bar{s})^4+\sum_{k=1}^{4}\rho_1^k\gamma_k}{(\rho_1-\bar{s})(1-\bar{s})^2(1+\rho_1-\bar{s})^2(1+\rho-\bar{s})^2}\right],
	\end{align}
	where $\bar{s}={s}/{\mu}$ and 
	\begin{align}\nonumber
	&\gamma_1\!=\!\rho^2_2(4\!-\!6\bar{s}+4\bar{s}^2-\bar{s}^3)\!+\!\rho_2(8-20\bar{s}+16\bar{s}^2-6\bar{s}^3+\bar{s}^4)+(1\!-\!\bar{s})^3(4\!-\!\bar{s}),
	\\&\nonumber
	\gamma_2=\rho^2_2(5-6\bar{s}+2\bar{s}^2)+\rho_2(12-22\bar{s}+14\bar{s}^2-3\bar{s}^3)+3(1-\bar{s})^2(2-\bar{s}),
	\\&\nonumber
	\gamma_3= \rho_2^2(2-\bar{s})+\rho_2(8-10\bar{s}+3\bar{s}^2)+3\bar{s}^2-7\bar{s}+4,
	\\&\nonumber
	\gamma_4=\rho_2(2-\bar{s})+1-\bar{s}.
	\end{align}
\end{theorem}
\begin{proof}
	The proof of Theorem \ref{MGFtheo1} appears in  Section \ref{MGF of the AoI under Policy 1}.
\end{proof}

\begin{theorem}\label{MGFtheo2}
	The MGF of the AoI of source 1 under the non-preemptive policy is given as
	\begin{align}\label{mnb102}
&M_{\Delta_1}(s)=\dfrac{\rho_1}{2\rho_1\rho_2+\rho+1}\times\left[\!\dfrac{\rho_2^3(1\!-\!\bar{s})^2\!+\!3\rho_2^2(1\!-\!\bar{s})^3\!+\!3\rho_2(1\!-\!\bar{s})^4\!+\!(1\!-\!\bar{s})^5\!+\!\sum_{k=1}^{3}\!\rho_1^k\bar{\gamma}_k}{(\rho_1-\bar{s})(1-\bar{s})^3(1+\rho_1-\bar{s})(1+\rho-\bar{s})}\!\right]
	\end{align}
	where $\bar{s}={s}/{\mu}$ and 
	\begin{align}\nonumber
	&\bar{\gamma}_1=\rho^3_2(3-3\bar{s}+\bar{s}^2)+\rho^2_2(9-19\bar{s}+11\bar{s}^2-2\bar{s}^3)+\\&\nonumber~~~~~\rho_2(9-28\bar{s}+29\bar{s}^2-11\bar{s}^3+\bar{s}^4)+3(1-\bar{s})^4,
	\\&\nonumber
	\bar{\gamma}_2=\rho^3_2(2-\bar{s})+\rho^2_2(8-10\bar{s}+3\bar{s}^2)+\rho_2(9-19\bar{s}+11\bar{s}^2-2\bar{s}^3)+(1-\bar{s})^3,
	\\&\nonumber
	\bar{\gamma}_3= \rho_2^2(2-\bar{s})+\rho_2(3-3\bar{s}+\bar{s}^2)+(1-\bar{s})^2.
	\end{align}
\end{theorem}
\begin{proof}
	The proof of Theorem \ref{MGFtheo2} appears in  Section \ref{MGF of the AoI under Policy 2}.
\end{proof}
{Note that the expressions are \emph{exact}; they characterize the MGF of the AoI in \emph{closed form}.}
\vspace{-.4cm}
\section{ The SHS Technique to Calculate MGF}\label{The SHS Technique to Calculate MGF}	
In the following, we briefly present how to use the SHS technique for our MGF analysis in Section \ref{AoI Analysis Using the SHS Technique}. We refer the readers to \cite{8469047,9103131} for more details.

The SHS technique models a queueing system  through the states $(q(t), \bold{x}(t))$, where  ${q(t)\in \mathcal{Q}=\{0,1,\ldots,m\}}$ is a continuous-time finite-state Markov
chain that  
describes the occupancy of the system and ${\bold{x}(t)=[x_0(t)\cdots x_n(t)]\in \mathbb{R}^{1\times(n+1)}}$ is a continuous
process that describes the evolution of age-related processes in the system. Following the approach in  \cite{8469047,moltafet2020sourceawareage}, we label the source of interest as source 1 and  employ the continuous process $ \bold{x}(t) $ to track the age of source 1 updates at the sink.

The Markov chain $q(t)$ can be presented as a graph $(\mathcal{Q},\mathcal{L})$ where 
each discrete state $q(t)\in \mathcal{Q}$ is a node of the chain and a (directed) link $ l\in\mathcal{L} $  from node $ q_l $ to node $q'_{l}$ indicates a transition from state $ {q_l \in \mathcal{Q}}$ to state ${q'_{l}\in \mathcal{Q}}$.

A transition occurs  when a packet arrives or departs in the system. Since the time elapsed between departures and arrivals is exponentially distributed, transition $l\in\mathcal{L}$ from state $ q_l $ to state $q'_{l}$ occurs with the  exponential rate $\lambda^{(l)}\delta_{q_l,q(t)}$,
where the Kronecker delta function $\delta_{q_l,q(t)}$ ensures that the transition $ l $ occurs only when the discrete
state $ q(t) $ is equal to $ q_l $.  When a transition $l$ occurs, the  discrete state $ q_l $ changes   to state $q'_{l}$, and the continuous state $\bold{x}$ is reset to $\bold{x}'$ according to a binary reset map matrix ${\bold{A}_l}\in\mathbb{B}^{(n+1)\times(n+1)}$ as ${\bold{x}'\!=\!\bold{x}\bold{A}_l}$. In addition, as long as  state $q(t)$ is unchanged we have  ${\dot{\bold{x}}(t)\!\triangleq\!\dfrac{\partial\bold{x}(t)}{\partial t}\!=\!\bold{1}}$, where $\bold{1}$ is the row vector $[1\cdots1]\!\in\! \mathbb{R}^{1\times(n+1)}$.

Note that unlike in a typical continuous-time Markov
chain, a transition
from a  state to itself  (i.e., a self-transition) is possible in  $q(t)\in \mathcal{Q}$. In the case of a self-transition, a reset of the continuous state $\bold{x}$ takes place, but the discrete state remains the same. In addition, for a given pair of states $\bar q,\hat q\in \mathcal{Q}$, there  may be multiple
transitions $ l $ and $ l' $ so that the discrete state changes from $ \bar q $ to
$ \hat q $ but the transition reset maps $\bold{A}_l $ and $ \bold{A}_{l'}$ are different (for more details, see  \cite[Section III]{8469047}). 

To calculate the MGF of the  AoI using the SHS technique, the state probabilities of the Markov chain,  the correlation vector between the discrete state $q(t)$ and the continuous state $\bold{x}(t)$, and  the correlation vector between the discrete state $q(t)$ and the exponential function $e^{s\bold{x}(t)},~s\in\mathbb{R}$, need to be defined. Let $\pi_q(t)$ denote the probability of being in state $q$ of the Markov chain. Let $\bold{v}_q(t)=[{v}_{q0}(t)\cdots{v}_{qn}(t)]\in\mathbb{R}^{1\times(n+1)}$  denote the correlation vector between the discrete state $q(t)$ and the continuous state $\bold{x}(t)$. Let $\bold{v}^s_q(t)=[{v}^s_{q0}(t)\cdots{v}^s_{qn}(t)]\in\mathbb{R}^{1\times(n+1)}$    denote the correlation vector between the state $q(t)$ and the exponential function $e^{s\bold{x}(t)}$. Accordingly, we have 
\begin{equation}
\pi_q(t)=\mathrm{Pr}(q(t)=q)=\mathbb{E}[\delta_{q,q(t)}], \,\,\,\forall q\in\mathcal{Q},
\end{equation}
\begin{equation}
\bold{v}_q(t)=[{v}_{q0}(t)\cdots{v}_{qn}(t)]=\mathbb{E}[\bold{x}(t)\delta_{q,q(t)}],\,\,\,\forall q\in\mathcal{Q},
\end{equation}
\begin{equation}
\bold{v}^s_q(t)=[{v}^s_{q0}(t)\cdots{v}^s_{qn}(t)]=\mathbb{E}[e^{s\bold{x}(t)}\delta_{q,q(t)}],\,\,\,\forall q\in\mathcal{Q}.
\end{equation}

Let $\mathcal{L}'_q$ denote the set of incoming transitions and  $\mathcal{L}_q$ denote the set of outgoing transitions for state $q$, defined as 
\begin{align}\nonumber
&\mathcal{L}'_q=\{l\in\mathcal{L}:q'_{l}=q\},~~~~\mathcal{L}_q=\{l\in\mathcal{L}:q_{l}=q\},\,\,\,\forall q\in\mathcal{Q}.
\end{align}
Following the ergodicity assumption of the Markov chain $q(t)$ in the AoI analysis \cite{8469047,9103131}, the state
probability vector $\boldsymbol{\pi}(t)=[\pi_0(t) \cdots \pi_m(t)]$ converges uniquely 
to the  stationary vector $\bar{\boldsymbol{\pi}}=[\bar{\pi}_0 \cdots \bar{\pi}_m]$ satisfying
\begin{align}\nonumber
\bar{{\pi}}_q\textstyle\sum_{l\in\mathcal{L}_q}\lambda^{(l)}=\textstyle\sum_{l\in\mathcal{L}'_q}\lambda^{(l)}\bar{{\pi}}_{q_l}, \,\forall q\in\mathcal{Q},~~
\textstyle\sum_{q\in\mathcal{Q}}\bar{{\pi}}_q=1.
\end{align}
Further,  it has been shown in \cite[Theorem 1]{9103131} that under the ergodicity assumption of the Markov chain $q(t)$, if we can find a non-negative limit ${\bar{\bold{v}}_q=[\bar{v}_{q0} \cdots \bar{v}_{qn} ], \forall q\in \mathcal{Q}}$, for the correlation vector ${\bold{v}_q(t)}$ satisfying
\begin{align}\label{asleq}	
\bar{\bold{v}}_q\textstyle\sum_{l\in\mathcal{L}_q}\lambda^{(l)}=\bar{{\pi}}_q\bold{1}+\textstyle\sum_{l\in\mathcal{L}'_q}\lambda^{(l)}\bar{\bold{v}}_{q_l}\bold{A}_l, \,\,\,\forall q\in\mathcal{Q},
\end{align}
there exists $s_0>0$ such that for all $s<s_0$, $\bold{v}^s_q(t), \forall q\in \mathcal{Q},$ converges to  $\bold{\bar v}^s_q$ that satisfies
\begin{align}\label{aslieq}	
\bold{\bar v}^s_q\textstyle\sum_{l\in\mathcal{L}_q}\lambda^{(l)}\!=\!s\bold{\bar v}^s_q\!+\!\textstyle\sum_{l\in\mathcal{L}'_q}\lambda^{(l)}[\bold{\bar v}^s_{q_l}\bold{A}_l\!+\!\bar{{\pi}}_{q_l}\bold{1}\bold{\hat A}_l], \,\,\,\forall q\in\mathcal{Q},
\end{align}
where $\bold{\hat A}_l\in\mathbb{B}^{(n+1)\times(n+1)}$ is a binary  matrix whose ${k,j}$th element, $\bold{\hat A}_l(k,j)$, is given as 
\begin{align}\nonumber
\bold{\hat A}_l(k,j)\!=\!\begin{cases}
1,\!&\!k\!=\!j, \text{and $j$th column of $ \bold{A}_l $ is a zero vector, } \\
0,\!&\!\text{otherwise.}
\end{cases}
\end{align}
Finally, the MGF of the state $\bold{x}(t)$, which is calculated by $\mathbb{E}[e^{s\bold{x}(t)}]$, converges to the stationary vector  {\cite[Theorem 1]{9103131}}
\begin{align}\label{AOIANAL}	
\mathbb{E}[e^{s\bold{x}}]=\textstyle\sum_{q\in\mathcal{Q}}\bold{\bar v}^s_q.
\end{align}
As \eqref{AOIANAL} implies, if the first element of continuous state $\bold{x}(t)$ represents the AoI of source 1 at the sink, the MGF of the AoI of source 1 at the sink converges to 
\begin{align}\label{MGFage}
M_{\Delta_1}(s)=\textstyle\sum_{q\in\mathcal{Q}}{\bar v}^s_{q0}.
\end{align}

From \eqref{MGFage},  the main challenge in calculating the MGF of the AoI of source 1 using the SHS technique reduces to deriving the first elements of correlation vectors $\bold{\bar v}^s_q,$ 
${\forall{q}\in\mathcal{Q}}$. 
\vspace{-.5cm}
\section{ AoI Analysis Using the SHS Technique}\label{AoI Analysis Using the SHS Technique}
 In this section, we use  the SHS technique to calculate the MGF of the AoI of source 1 in \eqref{MGFage} for each source under the self-preemptive and non-preemptive policies.

{The discrete state space is ${\mathcal{Q}=\{00,01,02,21,12\}},$ which is determined according to the occupancy of the system.} State $a_1a_2$ indicates that  a packet of source $ a_2$  is under service when $ a_2\ne0$ and a packet of source $ a_1 $ is in the waiting queue when $ a_1\ne0$. Note that ${a_1=0}$ (resp. $ {a_2=0}$) indicates that the waiting queue (resp. the server) is empty. {For example, state $a_1a_2=12$ indicates that a source 1 packet is in the waiting queue and a source 2 packet is under service.}
%
%
%

The continuous process is ${\bold{x}(t)=[x_0(t)~x_1(t)~x_2(t)]}$, where $x_0(t)$ is the current AoI of source 1 at time instant $t$, $\Delta_1(t)$;  $ x_1(t)  $ encodes  what the AoI of source 1 would become if the packet that is under service is delivered to the sink at time instant $t$; $ x_2(t)  $ encodes  what the AoI of source 1 would become if the  packet in the waiting queue is delivered to the sink at time instant $t$.

Recall that  to calculate the MGF of the  AoI of source 1 in \eqref{MGFage} we need to find ${\bar v}^s_{q0}, \forall q\in\mathcal{Q}$, which are the solution of  the system of linear equations \eqref{aslieq}	 with variables $\bold{\bar v}^s_q, \forall{q}\in\mathcal{Q}$. To form the system of linear equations \eqref{aslieq}, we need to determine  $\bar{{\pi}}_q$, $\bold{A}_l$, and   $\bold{\hat A}_l$   for each state $\forall{q}\in\mathcal{Q}$, and transition ${l\in\mathcal{L}'_q}$.
 Next, we derive these under the self-preemptive  and non-preemptive policies in Sections \ref{MGF of the AoI under Policy 1} and \ref{MGF of the AoI under Policy 2}, respectively.

\subsection{MGF of the AoI Under the Self-Preemptive Policy}\label{MGF of the AoI under Policy 1}
The Markov chain for the discrete state $q(t)$  is shown in Fig.~\ref{Chain_MGF}. The transitions between the discrete states ${{q_l \rightarrow q'_l}, \,\,\forall l\in \mathcal{L}}$, and their effects on the continuous state $\bold{x}(t)$ 
are summarized in Table \ref{Multi_MGF_Table_1}. The explanations of the transitions can be found in \cite[Section~IV.B]{moltafet2020sourceawareage}.

\begin{figure}
	\centering
	\includegraphics[scale=.8]{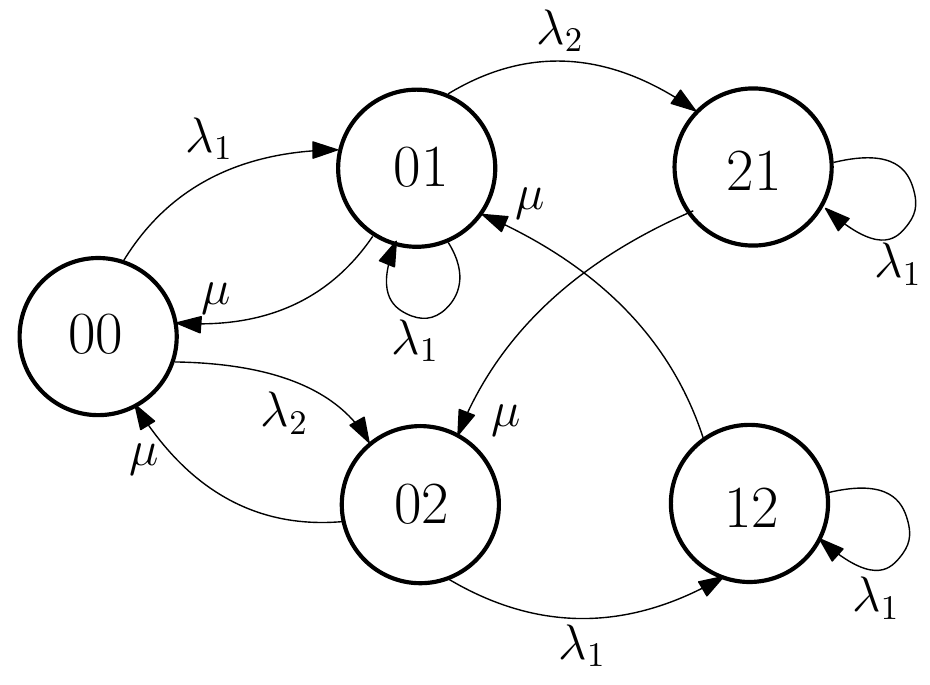}\vspace{-3mm}			
	\caption{The SHS Markov chain for the self-preemptive policy.}
	\vspace{-3mm}
	\label{Chain_MGF}
\end{figure}

\begin{table}
	\centering\small
	\caption{Table of transitions for  the self-preemptive policy}
	\label{Multi_MGF_Table_1}
	\begin{tabular}{ |l|l|c|c|c|c|c|}
		\hline
		\textit{l}  & $q_l \rightarrow q'_l $&$\lambda^{(l)}$& $\bold{x}\bold{A}_l$&$\bold{A}_l$&$\hat{\bold{A}}_l$ \\			
		\hline			
		1&$00 \rightarrow  01$& $\lambda_1$&$\left[x_0 ~ 0 ~ x_2 \right]$&$\tiny\begin{bmatrix}
		1 & 0 & 0\\
		0 & 0 & 0\\
		0 & 0 & 1
		\end{bmatrix}$&$\tiny\begin{bmatrix}
		0 & 0 & 0\\
		0 & 1 & 0\\
		0 & 0 & 0
		\end{bmatrix}$	\\		
		\hline
		2&$00 \rightarrow  02$&$\lambda_2$&$\left[x_0~ x_0~ x_2 \right]$&$\tiny\begin{bmatrix}
		1 & 1 & 0\\
		0 & 0 & 0\\
		0 & 0 & 1
		\end{bmatrix}$&$\tiny\begin{bmatrix}
		0 & 0 & 0\\
		0 & 0 & 0\\
		0 & 0 & 0
		\end{bmatrix}$\\
		\hline	
		3&$01 \rightarrow  01$&$\lambda_1$&$\left[x_0~ 0~ x_2 \right]$&$\tiny\begin{bmatrix}
		1 & 0 & 0\\
		0 & 0 & 0\\
		0 & 0 & 1
		\end{bmatrix}$&$\tiny\begin{bmatrix}
		0 & 0 & 0\\
		0 & 1 & 0\\
		0 & 0 & 0
		\end{bmatrix}$\\
		\hline	
		4&$01 \rightarrow  21$&$\lambda_2$&$\left[x_0~ x_1~ x_1 \right]$&$\tiny\begin{bmatrix}
		1 & 0 & 0\\
		0 & 1 & 1\\
		0 & 0 & 0
		\end{bmatrix}$&$\tiny\begin{bmatrix}
		0 & 0 & 0\\
		0 & 0 & 0\\
		0 & 0 & 0
		\end{bmatrix}$\\
		\hline	
		5&$02 \rightarrow 12 $&$\lambda_1$&$\left[x_0~ x_0~ 0 \right]$&$\tiny\begin{bmatrix}
		1 & 1 & 0\\
		0 & 0 & 0\\
		0 & 0 & 0
		\end{bmatrix}$&$\tiny\begin{bmatrix}
	    0 & 0 & 0\\
		0 & 0& 0\\
		0 & 0 & 1
		\end{bmatrix}$\\
		\hline	
		6&$21 \rightarrow  21$&$\lambda_1$&$\left[x_0~ 0~ 0 \right]$&$\tiny\begin{bmatrix}
		1 & 0 & 0\\
		0 & 0 & 0\\
		0 & 0 & 0
		\end{bmatrix}$&$\tiny\begin{bmatrix}
		0 & 0 & 0\\
		0 & 1 & 0\\
		0 & 0 & 1
		\end{bmatrix}$\\
		\hline	
		7&$12 \rightarrow 12 $&$\lambda_1$&$\left[x_0~ x_0~ 0 \right]$&$\tiny\begin{bmatrix}
		1 & 1 & 0\\
		0 & 0 & 0\\
		0 & 0 & 0
		\end{bmatrix}$&$\tiny\begin{bmatrix}
		0 & 0 & 0\\
		0 & 0 & 0\\
		0 & 0 & 1
		\end{bmatrix}$\\
		\hline	
		8&$01 \rightarrow  00$&$\mu$&$\left[x_1~ x_1~ x_2 \right]$&$\tiny\begin{bmatrix}
		0 & 0 & 0\\
		1 & 1 & 0\\
		0 & 0 & 1
		\end{bmatrix}$&$\tiny\begin{bmatrix}
		0 & 0 & 0\\
		0 & 0 & 0\\
		0 & 0 & 0
		\end{bmatrix}$\\
		\hline	
		9&$02 \rightarrow  00$&$\mu$&$\left[x_0~ x_1~ x_2 \right]$&$\tiny\begin{bmatrix}
		1 & 0 & 0\\
		0 & 1 & 0\\
		0 & 0 & 1
		\end{bmatrix}$&$\tiny\begin{bmatrix}
		0 & 0 & 0\\
		0 & 0 & 0\\
		0 & 0 & 0
		\end{bmatrix}$\\
		\hline	
		10&$21 \rightarrow  02$&$\mu$&$\left[x_1~ x_1~ x_2 \right]$&$\tiny\begin{bmatrix}
		0 & 0 & 0\\
		1 & 1 & 0\\
		0 & 0 & 1
		\end{bmatrix}$&$\tiny\begin{bmatrix}
		0 & 0 & 0\\
		0 & 0 & 0\\
		0 & 0 & 0
		\end{bmatrix}$\\
		\hline	
		11&$12 \rightarrow  01$&$\mu$&$\left[x_0~ x_2~ x_2 \right]$&$\tiny\begin{bmatrix}
		1 & 0 & 0\\
		0 & 0 & 0\\
		0 & 1 & 1
		\end{bmatrix}$&$\tiny\begin{bmatrix}
		0 & 0 & 0\\
		0 & 0 & 0\\
		0 & 0 & 0
		\end{bmatrix}$\\
		\hline			
	\end{tabular}
		\vspace{-3mm}
\end{table}

As it has been shown in \cite[Section~IV.B]{moltafet2020sourceawareage}, the stationary probabilities are given as 
\begin{align}\label{proeqq001}
\bar{\boldsymbol{\pi}}= \dfrac{1}{2\rho_1\rho_2+\rho+1}\left[1~~ \rho_1~~ \rho_2 ~~\rho_1\rho_2 ~~\rho_1\rho_2\right].
\end{align}
%

Recall from Section \ref{The SHS Technique to Calculate MGF} that to calculate the MGF of the AoI, first, we need to make sure whether  we can find  non-negative vectors ${\bar{\bold{v}}_q=[\bar{v}_{q0} \cdots \bar{v}_{qn} ], \forall q\in \mathcal{Q}}$,  satisfying \eqref{asleq}. 
As it is shown in \cite[Eq.~(13)]{moltafet2020sourceawareage}, the system of linear equations in \eqref{asleq} has a non-negative solution. Thus, the MGF of the AoI can be calculated by solving the system of linear equations in \eqref{aslieq}. We form \eqref{aslieq}
by  substituting  the values of $\bold{A}_l$ and   $\bold{\hat A}_l$ presented in Table \ref{Multi_MGF_Table_1} and  the vector $ \bar{\boldsymbol{\pi}}$ in   \eqref{proeqq001}. By solving the formed system of linear equations, the values of $\bar{v}^s_{q0}, \,\,\forall q\in\mathcal{Q}$, under the self-preemptive policy are calculated as presented in Appendix \ref{Valuesof  mgf vs 1 appendix}.

Finally,  substituting the values of $\bar{v}^s_{q0}, \,\,\forall q\in\mathcal{Q}$, 
into \eqref{MGFage}, we obtain the  MGF of the AoI of source 1 under  the self-preemptive policy, as given in Theorem \ref{MGFtheo1}.

\subsection{MGF of the AoI Under the Non-Preemptive Policy}\label{MGF of the AoI under Policy 2}
The Markov chain of the non-preemptive policy is the same as that for the self-preemptive policy. Thus, the stationary probability vector  $\bar{\boldsymbol{\pi}}$ of the non-preemptive policy is given in \eqref{proeqq001}.   The transitions between the discrete states ${{q_l \rightarrow q'_l}}$, and their effects on the  state $\bold{x}(t)$  for $l\in \{1,2,4,5,7,8,9,10,11\}$ are same as those fo the self-preemptive policy. The transitions $l\in \{3,6\}$ and their effects on the continuous state $\bold{x}(t)$ are summarized in Table~\ref{Multi_MGF_Table_2}. The explanations of the transitions can be found in \cite[Section~IV.C]{moltafet2020sourceawareage}.
 \begin{table}
	\centering\small
	\caption{Table of transitions for the non-preemptive policy}
	\label{Multi_MGF_Table_2}
	\begin{tabular}{ |l|l|c|c|c|c|}
		\hline
		\textit{l}  & $q_l \rightarrow q'_l $&$\lambda^{(l)}$& $\bold{x}\bold{A}_l$&$\bold{A}_l$&$\hat{\bold{A}}_l$  \\			
		\hline	
		3&$01 \rightarrow  01$&$\lambda_1$&$\left[x_0~ x_1~ x_2 \right]$&$\tiny\begin{bmatrix}
		1 & 0 & 0\\
		0 & 1 & 0\\
		0 & 0 & 1
		\end{bmatrix}$&$\tiny\begin{bmatrix}
		0 & 0 & 0\\
		0 & 0 & 0\\
		0 & 0 & 0
		\end{bmatrix}$\\
		\hline	
		6&$21 \rightarrow  21$&$\lambda_1$&$\left[x_0~ x_1~ x_1 \right]$&$\tiny\begin{bmatrix}
		1 & 0 & 0\\
		0 & 1 & 1\\
		0 & 0 & 0
		\end{bmatrix}$&$\tiny\begin{bmatrix}
		0 & 0 & 0\\
		0 & 0 & 0\\
		0 & 0 & 0
		\end{bmatrix}$\\
		\hline
	\end{tabular}
		\vspace{-5mm}
\end{table}

As it is shown in \cite[Eq.~(14)]{moltafet2020sourceawareage}, the system of linear equations in \eqref{asleq} has a non-negative solution. Thus, the MGF of the AoI can be calculated by solving the system of linear equations in \eqref{aslieq}. We form \eqref{aslieq}
by  substituting  the values of $\bold{A}_l$ and   $\bold{\hat A}_l$ presented in Tables \ref{Multi_MGF_Table_1} and \ref{Multi_MGF_Table_2} and  the vector $ \bar{\boldsymbol{\pi}}$ in   \eqref{proeqq001}.
%
%
 By solving the formed system of linear equations, the values of $\bar{v}^s_{q0}, \,\,\forall q\in\mathcal{Q}$, under the non-preemptive policy are calculated as presented in Appendix \ref{Valuesof  mgf vs 2 appendix}. 

Finally,  substituting the values of $\bar{v}^s_{q0}, \,\,\forall q\in\mathcal{Q}$, 
into \eqref{MGFage} results in the  MGF of the AoI of source 1 under the non-preemptive policy, given in Theorem \ref{MGFtheo2}.

The following remark presents how we can derive different moments of the AoI by using the MGF.
\normalfont \begin{rem}\label{rem1MGF}\normalfont
The $i$th moment of the  AoI of source 1, $\Delta_1^{(i)}$, is calculated as
\begin{align}
\Delta_1^{(i)}=\dfrac{\mathrm{d}^i(M_{\Delta_1}(s))}{\mathrm{d}s^i}\Big|_{s=0}.
\end{align}
\end{rem}
\section{First and Second Moments of the AoI}\label{Numerical Results}
Having derived the MGF of the AoI presented in Theorems~\ref{MGFtheo1} and \ref{MGFtheo2}, we apply Remark \ref{rem1MGF} to derive the first and second moments of the AoI of source 1.
{It is easy to show that the first moment of the AoI of source 1 under the self-preemptive and non-preemptive policies coincide with the results in \cite[Theorems~2 and 3]{moltafet2020sourceawareage}, as expected.}
%
%
The second moment of the AoI of source 1 under  the self-preemptive policy is given as 
\begin{align}\label{secondmoment1}
\Delta^{(2)}_1=\dfrac{2{(\rho_2+1)}^3+2\sum_{k=1}^8 \rho_1^k \xi_k}
{\mu \rho^2_1\left(1+\rho_1\right)^3(1+\rho)^2(2\rho_1\rho_2+\rho+1)},
\end{align}
where
$
\xi_1\!=\!7\rho_2^3\!+\!21\rho_2^2\!+\!21\rho_2\!+\!7,$ $
\xi_2\!=\!22\rho_2^3\!+\!68\rho_2^2\!+\!68\rho_2\!+\!22,$ $
\xi_3\!=\!40\rho_2^3\!+\!113\rho_2^2\!+\!134\rho_2\!+\!41$ $
\xi_4\!=\!36\rho_2^3\!+\!161\rho_2^2\!+\!180\rho_2\!+\!50,$ $
\xi_5\!=\!18\rho_2^3\!+\!113\rho_2^2\!+\!160\rho_2\!+\!41,$ $
\xi_6\!=\!4\rho_2^3\!+\!45\rho_2^2\!+\!88\rho_2\!+\!22,$ $
\xi_7=8\rho_2^2+28\rho_2+7,$ $
\xi_8=4\rho_2+1.
$
The second moment of the AoI of source 1 under the non-preemptive policy is 
\begin{align}\label{secondmoment2}
\Delta^{(2)}_1\!=\!\dfrac{2{(\rho_2+1)}^5+2\sum_{k=1}^7 \rho_1^k \hat\xi_k}
{\mu \rho^2_1\left(1\!+\!\rho_1\right)^2\left(1\!+\!\rho_2\right)^2(1\!+\!\rho)^2(2\rho_1\rho_2\!+\!\rho\!+\!1)},
\end{align}
where
$
\hat\xi_1=6\rho_2^5+30\rho_2^4+60\rho_2^3+60\rho_2^2+30\rho_2+6,$ $
\hat\xi_2=18\rho_2^5+91\rho_2^4+182\rho_2^3+180\rho_2^2+88\rho_2+17,$ $
\hat\xi_3\!=\!34\rho_2^5\!+\!178\rho_2^4\!+\!361\rho_2^3\!+355\rho_2^2\!+\!169\rho_2\!+\!31,$ $
\hat\xi_4\!=\!29\rho_2^5\!+\!190\rho_2^4\!+439\rho_2^3\!+463\rho_2^2\!+\!224\rho_2+39,$ $
\hat\xi_5=9\rho_2^5+97\rho_2^4+293\rho_2^3+365\rho_2^2+192\rho_2+32,$ $
\hat\xi_6=18\rho_2^4+92\rho_2^3+151\rho_2^2+93\rho_2+15,$ $
\hat\xi_7=9\rho_2^3+24\rho_2^2+19\rho_2+3.
$

Fig. \ref{SDV}  depicts the average AoI of source 1 and its standard deviation ($ \sigma $) as a function of $\lambda_1$ under the two packet management policies for $\mu=1$ and $\lambda=5$. This figure shows that in a status update system, the standard deviation of the AoI might have a large value. Thus, to have a reliable system, in addition to optimizing the average AoI, we need to take the  higher moments of the AoI into account.

\begin{figure}
	\centering
	\includegraphics[scale=.54]{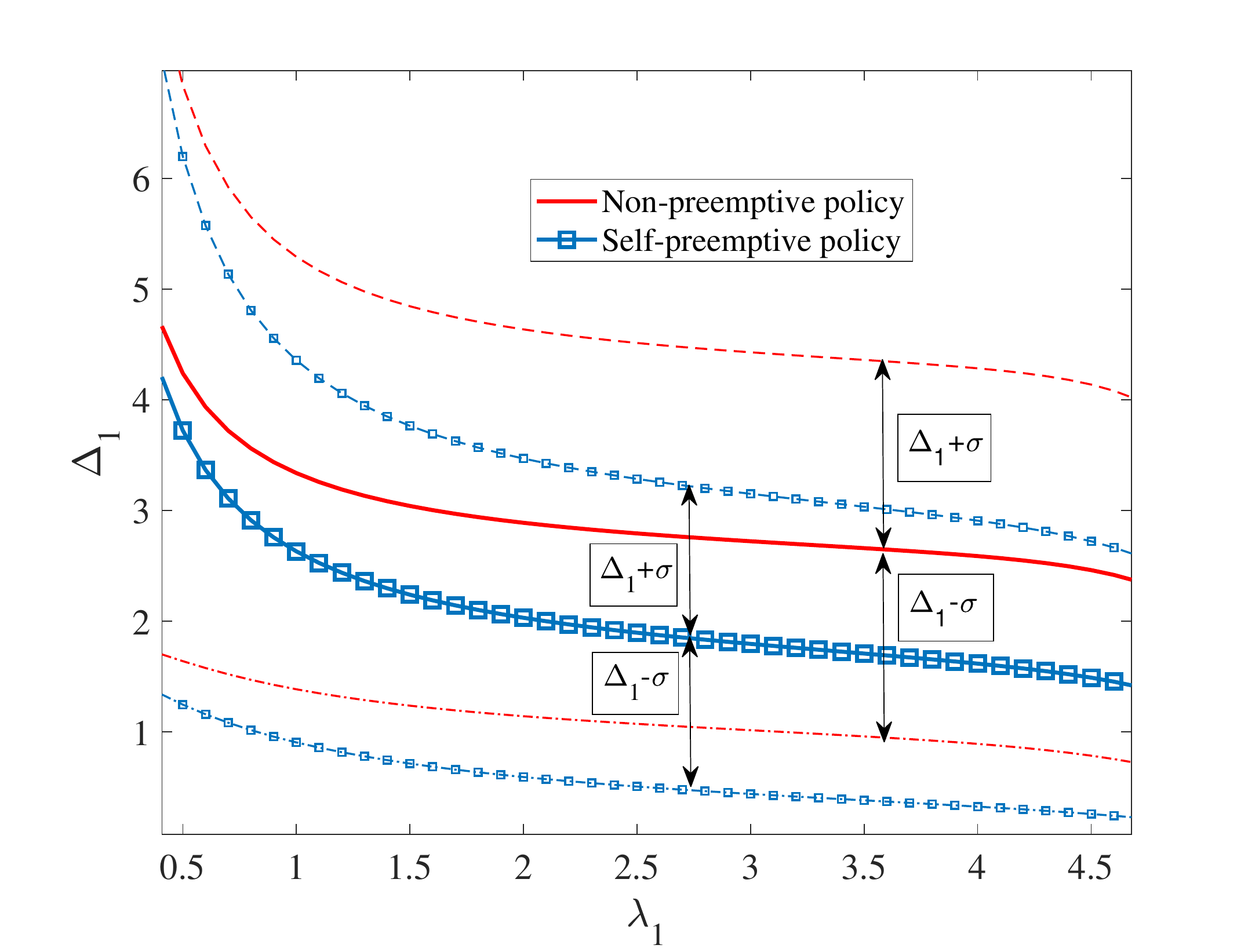}\vspace{-2mm}			
	\caption{The average AoI of source 1 and its standard deviation as a function of $\lambda_1$ under the two packet management policies for $\mu=1$ and $\lambda=5$. The dashed lines visualize the standard deviation of the AoI as $\Delta_1+\sigma$ and $\Delta_1-\sigma$.}
	\vspace{-8mm}
	\label{SDV}
\end{figure}

\section{Conclusion}\label{Conclusions}
We considered a status update system consisting of two independent sources and one server in which packets of each source are generated according to the Poisson process and packets in the system are served according to an exponentially distributed service time. We derived the MGF of the AoI under two packet management policies by using the SHS technique. We derived the first and second moments of the AoI by using the MGF and demonstrated the importance of considering higher moments in the AoI optimization. 

{The interesting future works would be to i) derive the MGF of the AoI for the packet management policy  where the packet of one source is  preempted by a new arrival of the same source whereas the arriving packets of the other source are discarded when there is a packet of that source in the system; and ii) study the average AoI minimization with a maximum variance constraint.}

\appendices
\section{Values of $\bar{v}^s_{q0}$ for  the Self-Preemptive Policy} \label{Valuesof  mgf vs 1 appendix}
\vspace{-5mm}
\begin{align}\nonumber
&\bar{v}^s_{00}=\dfrac{\rho_1}{2\rho_1\rho_2+\rho+1}\Big[\dfrac{(1-\bar{s})^2+\rho_2}{(\rho_1-\bar{s})(1+\rho_1-\bar{s})^2(1+\rho-\bar{s})^2}\\&\nonumber
+\dfrac{\rho_1^3+\rho_1^2(3-2\bar{s}+\rho_2)+\rho_1((2-\bar{s})^2+\rho_2(2-\bar{s})-1)}{(\rho_1-\bar{s})(1+\rho_1-\bar{s})^2(1+\rho-\bar{s})^2}\Big].
\end{align}
\begin{align}\nonumber
&\bar{v}^s_{10}=\dfrac{\rho^2_1}{2\rho_1\rho_2+\rho+1}
\Big[\dfrac{\rho_1^3(\rho_2+1-\bar{s})+\rho_1^2\alpha_{1,1}+\rho_1\alpha_{1,2}+(1-\bar{s})^3}
{(1-\bar{s})(\rho_1-\bar{s})(1+\rho_1-\bar{s})^2(1+\rho_2-\bar{s})(1+\rho-\bar{s})}+\\&\nonumber
\dfrac{(\rho_2+1)^2-3\rho_2\bar{s}-1}
{(1-\bar{s})(\rho_1-\bar{s})(1+\rho_1-\bar{s})^2(1+\rho_2-\bar{s})(1+\rho-\bar{s})}\Big],
\end{align}
where
$
\alpha_{1,1}=(\rho_2+2)^2+2(\bar{s}-2)^2+3(1-\rho_2)-9,$ and $
\alpha_{1,2}=\rho_2^2(2-\bar{s})+\rho_2(5-6\bar{s}+2\bar{s}^2)+3-7\bar{s}+5\bar{s}^2-\bar{s}^3.
$
\begin{align}\nonumber
&\bar{v}^s_{20}=\dfrac{\rho_1\rho_2}{2\rho_1\rho_2+\rho+1}\Big[\dfrac{1-2\bar{s}+\rho_2}
{(\rho_1-\bar{s})(1+\rho_1-\bar{s})^2(1+\rho-\bar{s})}+\\&\nonumber
\dfrac{\rho_1^3+\rho_1^2(3-2\bar{s}+\rho_2)+\rho_1(3(1-\bar{s})+\bar{s}^2+(2-\bar{s}))}
{(\rho_1-\bar{s})(1+\rho_1-\bar{s})^2(1+\rho-\bar{s})}\Big].
\end{align}
\begin{align}\nonumber
&\bar{v}^s_{30}=\dfrac{\rho^2_1\rho_2}{2\rho_1\rho_2+\rho+1}
\Big[\dfrac{\rho_1^3(\rho_2+1-\bar{s})+\rho_1^2\alpha_{3,1}+\rho_1\alpha_{3,2}+(\rho_2+1)^2}
{(\rho_1-\bar{s})(1-\bar{s})^2(1+\rho_1-\bar{s})^2(1+\rho_2-\bar{s})(1+\rho-\bar{s})}+\\&\nonumber\dfrac{(1-\bar{s})^3-3\rho_2\bar{s}-1}
{(\rho_1-\bar{s})(1-\bar{s})^2(1+\rho_1-\bar{s})^2(1+\rho_2-\bar{s})(1+\rho-\bar{s})}\Big],
\end{align}
where
$
\alpha_{3,1}=(\rho_2+2)^2+2(\bar{s}-1)^2-\bar{s}(3\rho_2+1)-3,$ and $
\alpha_{3,2}=\rho_2^2(2-\bar{s})+\rho_2(5-6\bar{s}+2\bar{s}^2)+3-7\bar{s}+5\bar{s}^2-\bar{s}^3.
$
\begin{align}\nonumber
&\bar{v}^s_{40}=\dfrac{\rho^2_1\rho_2}{2\rho_1\rho_2+\rho+1}
\Big[\dfrac{\rho_1^3+\rho_1^2(3-2\bar{s}+\rho_2)+\rho_1(3(1-\bar{s})}
{(\rho_1\!-\!\bar{s})(1\!-\!\bar{s})(1\!+\!\rho_1\!-\!\bar{s})^2(1\!+\!\rho-\bar{s})}\\&\nonumber+\dfrac{\rho_2(2-\bar{s})+\bar{s}^2+1)+1-2\bar{s}+\rho_2}
{(\rho_1-\bar{s})(1-\bar{s})(1+\rho_1-\bar{s})^2(1+\rho-\bar{s})}\Big].
\end{align}
\section{Values of $\bar{v}^s_{q0}$ for the Non-Preemptive Policy} \label{Valuesof  mgf vs 2 appendix}
\vspace{-5mm}
\begin{align}\nonumber
&\bar{v}^s_{00}=\dfrac{\rho_1}{2\rho_1\rho_2+\rho+1}
\Big[\dfrac{\rho_1^2(1-\bar{s})(1+\rho_2)+\rho_1\bar\alpha_{0,1}+(1+\rho_2)^2}
{(\rho_1-\bar{s})(1-\bar{s})(1+\rho_1-\bar{s})(1+\rho_2-\bar{s})(1+\rho-\bar{s})}+\\&\nonumber
\dfrac{\rho_2\bar{s}(\bar{s}-3)+(1-\bar{s})^3-1}
{(\rho_1-\bar{s})(1-\bar{s})(1+\rho_1-\bar{s})(1+\rho_2-\bar{s})(1+\rho-\bar{s})}\Big],
\end{align}
where 
$
\bar\alpha_{0,1}=\rho_2(\rho_2+3)+\rho_2\bar{s}(\bar{s}-3)+2(1-\bar{s}).
$
\begin{align}\nonumber
&\bar{v}^s_{10}=\dfrac{\rho^2_1}{2\rho_1\rho_2+\rho+1}
\Big[\dfrac{\rho_1^2\bar\alpha_{1,1}+\rho_1\bar\alpha_{1,2}+\rho_2^3+\rho_2^2(3-4\bar{s})}
{(\rho_1-\bar{s})(1-\bar{s})^2(1+\rho_1-\bar{s})(1+\rho_2-\bar{s})^2(1+\rho-\bar{s})}+\\&\nonumber
\dfrac{\rho_2(3-8\bar{s}+6\bar{s}^2-\bar{s}^3)+(1-\bar{s})^4}
{(\rho_1-\bar{s})(1-\bar{s})^2(1+\rho_1-\bar{s})(1+\rho_2-\bar{s})^2(1+\rho-\bar{s})}\Big],
\end{align}
where 
$\bar\alpha_{1,1}=(\rho_2+1)^2+\rho_2\bar{s}(\bar{s}-2)+(1-\bar{s})^2-1,$ and $
\bar\alpha_{1,2}=\rho_2^3+\rho_2^2(4-3\bar{s})+\rho_2(5-9\bar{s}+4\bar{s}^2-\bar{s}^3)+2(1-\bar{s})^3.
$
\begin{align}\nonumber
&\bar{v}^s_{20}=\dfrac{\rho_1\rho_2}{2\rho_1\rho_2+\rho+1}
\Big[\dfrac{\rho_1^2(\rho_2+1)+\rho_1(\rho_2^2+3\rho_2+2-\bar{s}(2\rho_2+3))+\rho_2^2}
{(\rho_1-\bar{s})(1-\bar{s})(1+\rho_1-\bar{s})(1+\rho_2-\bar{s})(1+\rho-\bar{s})}+\\&\nonumber\dfrac{\rho_2(2-3\bar{s})+1-3\bar{s}+2\bar{s}^2}
{(\rho_1-\bar{s})(1-\bar{s})(1+\rho_1-\bar{s})(1+\rho_2-\bar{s})(1+\rho-\bar{s})}\Big],
\end{align}
\vspace{-2mm}
\begin{align}\nonumber
&\bar{v}^s_{30}=\dfrac{\rho^2_1\rho_2}{2\rho_1\rho_2+\rho+1}
\Big[\dfrac{\rho_1^2\bar\alpha_{3,1}+\rho_1\bar\alpha_{3,2}+\rho_2^3+\rho_2^2(3-4\bar{s})}
{(\rho_1-\bar{s})(1-\bar{s})^3(1+\rho_1-\bar{s})(1+\rho_2-\bar{s})^2(1+\rho-\bar{s})}+\\&\nonumber\dfrac{\rho_2(3-8\bar{s}+6\bar{s}^2-\bar{s}^3)+(1-\bar{s})^2}
{(\rho_1-\bar{s})(1-\bar{s})^3(1+\rho_1-\bar{s})(1+\rho_2-\bar{s})^2(1+\rho-\bar{s})}\Big],
\end{align}
where 
$
\bar\alpha_{3,1}=(\rho_2+1)^2+\rho_2\bar{s}(\bar{s}-2)+(1-\bar{s})^2-1,$ and $
\bar\alpha_{3,2}=\rho_2^3+\rho^2_2(4-3\bar{s})+\rho_2(5-9\bar{s}+4\bar{s}^2-\bar{s}^3)+2(1-\bar{s})^3.
$
\begin{align}\nonumber
&\bar{v}^s_{40}=\dfrac{\rho^2_1\rho_2}{2\rho_1\rho_2+\rho+1}
\Big[\dfrac{\rho_1^2(\rho_2+1)+\rho_1(\rho_2^2+3\rho_2+2-\bar{s}(2\rho_2+3))}
{(\rho_1-\bar{s})(1-\bar{s})^2(1+\rho_1-\bar{s})(1+\rho_2-\bar{s})(1+\rho-\bar{s})}+\\&\nonumber\dfrac{(\rho_2+1)^2+3\rho_2-3\bar{s}+2\bar{s}^2}
{(\rho_1-\bar{s})(1-\bar{s})^2(1+\rho_1-\bar{s})(1+\rho_2-\bar{s})(1+\rho-\bar{s})}\Big].
\end{align}

\bibliographystyle{IEEEtran}
\bibliography{Bibliography}

\begin{thebibliography}{10}
\providecommand{\url}[1]{#1}
\csname url@samestyle\endcsname
\providecommand{\newblock}{\relax}
\providecommand{\bibinfo}[2]{#2}
\providecommand{\BIBentrySTDinterwordspacing}{\spaceskip=0pt\relax}
\providecommand{\BIBentryALTinterwordstretchfactor}{4}
\providecommand{\BIBentryALTinterwordspacing}{\spaceskip=\fontdimen2\font plus
\BIBentryALTinterwordstretchfactor\fontdimen3\font minus
  \fontdimen4\font\relax}
\providecommand{\BIBforeignlanguage}[2]{{%
\expandafter\ifx\csname l@#1\endcsname\relax
\typeout{** WARNING: IEEEtran.bst: No hyphenation pattern has been}%
\typeout{** loaded for the language `#1'. Using the pattern for}%
\typeout{** the default language instead.}%
\else
\language=\csname l@#1\endcsname
\fi
#2}}
\providecommand{\BIBdecl}{\relax}
\BIBdecl

\bibitem{6195689}
S.~Kaul, R.~Yates, and M.~Gruteser, ``Real-time status: How often should one
  update?'' in \emph{Proc. IEEE Int. Conf. on Computer. Commun. (INFOCOM)},
  Orlando, FL, USA, Mar. 25--30, 2012, pp. 2731--2735.

\bibitem{8469047}
R.~D. {Yates} and S.~K. {Kaul}, ``The age of information: Real-time status
  updating by multiple sources,'' \emph{{IEEE} Trans. Inform. Theory}, vol.~65,
  no.~3, pp. 1807--1827, Mar. 2019.

\bibitem{9099557}
M.~{Moltafet}, M.~{Leinonen}, and M.~{Codreanu}, ``On the age of information in
  multi-source queueing models,'' \emph{{IEEE} Trans. Commun.}, vol.~68, no.~8,
  pp. 5003--5017, May 2020.

\bibitem{9103131}
R.~D. {Yates}, ``The age of information in networks: Moments, distributions,
  and sampling,'' \emph{{IEEE} Trans. Inform. Theory}, vol.~66, no.~9, pp.
  5712--5728, Sep. 2020.

\bibitem{8437591}
S.~K. {Kaul} and R.~D. {Yates}, ``Age of information: Updates with priority,''
  in \emph{Proc. IEEE Int. Symp. Inform. Theory}, Vail, CO, USA, Jun. 17--22,
  2018, pp. 2644--2648.

\bibitem{8406966}
R.~D. {Yates}, ``Age of information in a network of preemptive servers,'' in
  \emph{Proc. IEEE Int. Conf. on Computer. Commun. (INFOCOM)}, Honolulu, HI,
  USA, Apr. 15--19, 2018, pp. 118--123.

\bibitem{8437907}
------, ``Status updates through networks of parallel servers,'' in \emph{Proc.
  IEEE Int. Symp. Inform. Theory}, Vail, CO, USA, Jun. 17--22, 2018, pp.
  2281--2285.

\bibitem{9013935}
A.~{Javani}, M.~{Zorgui}, and Z.~{Wang}, ``Age of information in multiple
  sensing,'' in \emph{Proc. IEEE Global Telecommun. Conf.}, Waikoloa, HI, USA,
  USA, Dec. 9--13, 2019.

\bibitem{9048914}
S.~{Farazi}, A.~G. {Klein}, and D.~{Richard Brown}, ``Average age of
  information in multi-source self-preemptive status update systems with packet
  delivery errors,'' in \emph{Proc. Annual Asilomar Conf. Signals, Syst.,
  Comp.}, Pacific Grove, CA, USA, USA, 2019.

\bibitem{9007478}
A.~{Maatouk}, M.~{Assaad}, and A.~{Ephremides}, ``On the age of information in
  a {CSMA} environment,'' \emph{IEEE/ACM Trans. Net.}, vol.~28, no.~2, pp.
  818--831, Feb. 2020.

\bibitem{moltafet2020sourceawareage}
\BIBentryALTinterwordspacing
M.~Moltafet, M.~Leinonen, and M.~Codreanu, ``Average {AoI} in multi-source
  systems with source-aware packet management,'' \emph{IEEE Trans. Commun},
  Early Access, 2020. [Online]. Available:
  \url{https://arxiv.org/pdf/2001.03959.pdf}
\BIBentrySTDinterwordspacing

\bibitem{9174099}
M.~{Moltafet}, M.~{Leinonen}, and M.~{Codreanu}, ``Average age of information
  for a multi-source {M/M/1} queueing model with packet management,'' in
  \emph{Proc. IEEE Int. Symp. Inform. Theory}, Los Angeles, CA, USA, Jun.
  21--26, 2020, pp. 1765--1769.

\end{thebibliography}
\end{document}